\documentclass[conference,a4paper]{IEEEtran}
\def\num{1} 

\IEEEoverridecommandlockouts
\usepackage[left=1.59cm, right=1.59cm, top=1.8cm, bottom = 2.75cm]{geometry} 
\usepackage{hyperref}
\usepackage{cite}
\usepackage{amsmath,amssymb,amsfonts}
\usepackage{amsthm} 
\usepackage[ruled, vlined]{algorithm2e}
\usepackage{graphicx}
\usepackage{textcomp}
\usepackage{xcolor}
\usepackage{eurosym}
\usepackage{dsfont}
\usepackage[tableposition=top]{caption}
\usepackage{bm}
\usepackage{newtxmath} 
\DeclareCaptionFormat{myformat}{#3}
\newtheorem{prop}{Proposition}
\newtheorem{defi}{Definition}

\def\BibTeX{{\rm B\kern-.05em{\sc i\kern-.025em b}\kern-.08em
    T\kern-.1667em\lower.7ex\hbox{E}\kern-.125emX}}
    
\begin{document}

\title{Online smart charging algorithm with asynchronous electric vehicles demand}

\author{\IEEEauthorblockN{Benoît Sohet}
\IEEEauthorblockA{benoit.sohet@edf.fr
\\ \textit{LIA, Avignon Univ.}\\\textit{EDF R\&D, MIRE Dept,}\\ 
\textit{Paris-Saclay, France}
}
\and

\IEEEauthorblockN{Yezekael Hayel}
\IEEEauthorblockA{
\{yezekael.hayel\hfill\\\hfill@univ-avignon.fr\}\\
\textit{LIA, Avignon Univ.}\\ \textit{Avignon, France}
}
\and

\IEEEauthorblockN{Olivier Beaude \hspace{4mm} Jean-Baptiste Bréal}
\IEEEauthorblockA{\{olivier.beaude, jean-baptiste.breal\}@edf.fr\\
\textit{EDF R\&D, OSIRIS and MIRE Dept,}\\ 
\textit{Paris-Saclay, France}
} 
\and


\and
\IEEEauthorblockN{Alban Jeandin}
\IEEEauthorblockA{alban.jeandin@izivia.com\\
\textit{Izivia, EDF group}\\ 
\textit{Courbevoie, France}
}
}

\maketitle

\begin{abstract}
The increasing penetration of Electric Vehicles (EVs) and renewable energies into the grid necessitates tools to smooth the demand curve.
To this end, this paper suggests an EV charging scheduling algorithm and a smart charging price. 
As EVs arrive at the charging station and leave at different times, the operator of the station applies at each EV arrival an online scheduling algorithm based on the concept of ``water filling''.
The EV charging price is guaranteed at their arrival and defined as a function of the online algorithm's output, following the idea of locational marginal pricing.
A numerical comparison with the offline version of the algorithm -- in which the operator knows in advance all future arrival and departure times -- shows the efficiency of the suggested online scheduling algorithm.

\end{abstract}

\section{Introduction}

Even though Electric Vehicles (EVs) represent an encouraging answer to local pollution (air quality, noise), they also bring challenges to the grid. In France for example,~\cite{RTE} expects a power demand increase of 2.2 to 3.6 GW by 2035.
Currently, the prononced local penetrations of EVs\footnote{More than 80,000 EVs in circulation in Paris region: \url{https://www.statistiques.developpement-durable.gouv.fr/sites/default/files/2020-04/immatriculations\_neuves\_2019.zip}.} may already lead to infrastructure investment costs.
On the electricity generation side, the dramatic increase of intermittent and distributed renewable energies expected in the near future~\cite{iea20} will make the whole electricity generation system less flexible in order to maintain the supply-demand balance.
For all these reasons, the flexibility of the electricity demand system and in particular of EV charging needs to be fully exploited, from smart charging to demand response.

Smart charging consists in postponing the EV charging profile in time compared to the plug and charge method, and has been largely studied~\cite{Wang2016, nimalsiri2019survey}.
Demand response mechanisms act on EV users' decisions about when, where and how much to charge by using price incentives, and also benefits from a large literature~\cite{vardakas2014survey, jordehi2019optimisation}.
In this work, the goal is to provide a smart charging algorithm and a price incentive tractable enough to be integrated in a complex electrical-transportation coupled system, which takes into account both the interactions between EV users (while driving and charging), and the different system operators~\cite{sohet21}.

In this system, the operator of an EV Charging Station (EVCS) is responsible for the EV charging scheduling and uses a centralized smart charging algorithm.
A natural example of such an algorithm is the Water Filling algorithm~\cite{shinwari2012water}, which reduces the variance in time of the total load at the EVCS in an efficient manner.
In the present paper, EVs can arrive at the EVCS and leave at different times, and the operator does not necessarily have this information in advance.
There exist papers such as~\cite{he2012optimal} which deal with this asynchronous charging need using an online charging scheduling, where the operator solves an optimization problem at each time slot with the freshly available EV information, but no simple explicit solution is given.

The Charging Unit Price (CUP) considered in this work is the Locational Marginal Pricing~\cite{li2013distribution} (LMP), where EVs pay the charging quantity multiplied by the marginal operator's cost associated with an additional marginal charging quantity.
Such a pricing scheme is known to be the best one to incite EVs to reduce the operator's cost~\cite{ALIZADEH17} (when EVs are faced with a choice, e.g., the charging quantity or the presence time at the EVCS).
However, the LMP requires a continuous real-time communication between EVs and the operator.
This assumption can be relaxed using for example the day-ahead real-time pricing~\cite{doostizadeh2012day}, where the CUP is fixed the day before.
For the moment, charging services such as the charging quantity and price should be guaranteed by the operator to the EVs when they arrive at the EVCS.

The main contributions of the present paper are:

\begin{itemize}
    \item The procedure the operator of the EVCS needs to follow in order to solve the charging scheduling problem with asynchronous EV arrival and departure times.
    This procedure includes at each arrival of EVs at the EVCS an explicit online algorithm based on WF and which does not require any optimization computations.
    \item An incentive CUP corresponding to the marginal operator's costs, costs which are the results of the online charging scheduling algorithms.
    This CUP is communicated to EVs at their arrival.
    \item Comparison of the online charging schedule and CUP with those obtained in the optimal case of an omniscient operator (or offline charging problem), using real data of EVs arrival and departure times and PhotoVoltaic (PV) electricity generation.
\end{itemize}

The paper is organized as follows. 
The framework and notations are introduced in Sec.~II. The offline and charging scheduling problem are described resp. in Sec.~III and~IV.
The CUP is defined in Sec.~V and Sec.~VI corresponds to the numerical studies.
Finally, conclusions and perspectives are given in last section.


\noindent
\textbf{Notations:} vectors can be written $\{v_1, .., v_N\}$ or in bold $\bm{v}$.
\section{Charging scheduling framework}

The operator of an EVCS wants to determine the charging profile of EVs plugged at its EVCS during time period $\mathcal{T}$.
This period $\mathcal{T} = \{1\,,\dots\,,T\}$ is assumed to be discretized in $T$ time slots of equal duration $\delta$.
The EVs are divided into different classes $(a,d)$, depending on their arrival $a$ and departure $d$ time slots at the EVCS.
More precisely, EVs of class $(a,d)$ arrive at the EVCS at the beginning of time slot $a\in\mathcal{T}$ and leave at the end of time slot $d\in\mathcal{T}$, and therefore may only charge during time slots $\{a,\dots,d\}$.
For example if $d=a$, the corresponding EV class only charges during time slot $a$.
The set of times' pairs $(a,d)$ is written $\mathcal{R}\subseteq \mathcal{T}\times\mathcal{T}$.
The total charging need aggregated over all EVs of class $(a,d)$ is written $L^{(a,d)}$.

For each class $(a,d)$, the operator wants to determine the charging power $\ell^{(a,d)}_{t}$ at each time slot $t\in\{a,\dots,d\}$ aggregated over all EVs of class $(a,d)$, so that the corresponding aggregated charging need $L^{(a,d)}$ is fulfilled at departure time slot $d$, i.e. $\sum_{t=a}^{d} \ell^{(a,d)}_{t} = L^{(a,d)}/\delta$.
For all $t$, $\ell^{(a,d)}_{t}\geq 0$ but note that vehicle to grid could constitute a direct extension of this work by relaxing this constraint.
This work does not focus on how a charging power aggregated over an EV class is distributed between the EVs of this class, which is assumed feasible. 

The per-class aggregated charging profile selected by the operator (vector written $\tilde{\bm{\ell}}$) is the one minimizing some charging cost function (whose minimum value is written $\tilde{G}$), which also depends on electrical usages $\bm{\ell}^0$ at the EVCS other than EV charging and called nonflexible (typically, the electrical consumption of a tertiary site if considering an EVCS at a work place).
This nonflexible part can also include a local generation, e.g. when a PV panel is associated to the EVCS.
Note that the nonflexible term can be then negative, meaning that there is more local electricity generation at the EVCS than consumption.
Even if this term can include both generation and consumption, it will be simply called ``consumption" in the following to get a generic terminology.
The charging cost function depends on the cost of the total power load $\ell^0_t + \sum_{(a,d)}\ell^{(a,d)}_t$ at a given time slot $t$, represented by an increasing and convex function $f$.
As commonly used in the literature, this function can represent EVCS on-site economic mechanisms\footnote{See French network tariff: \url{https://www.enedis.fr/sites/default/files/TURPE_5bis_plaquette_tarifaire_aout_2020.pdf} (in French).} (charging bill, including an incentive associated to self-consumption) or local network congestion effects~\cite{mohsenian10} (losses, voltage regulation, equipment aging).

The next two sections introduce two different charging scheduling problems, depending on the information available to the operator.
Note however that in both problems, the operator is assumed to know in advance (before the first charging time slot) the nonflexible consumption $\bm{\ell}^0$.

 \section{Offline optimization problem}
 \label{sec:offline}
 
In this section, the operator is assumed to know in advance all arrival $a$ and departure $d$ time slots and the corresponding charging needs $L^{(a,d)}$ before the beginning of the whole time period $\mathcal{T}$.
In practice, all EVs could declare this information through an app before the first charging time slot, or the operator could base the values $L^{(a,d)}$ on statistical data.
Therefore, the operator can compute the optimal charging profiles \textit{offline}, i.e. before the beginning of $\mathcal{T}$, by solving the following charging scheduling problem~\ref{eq:offline}:

\begin{equation}
\begin{aligned}
&\min_{\hspace{7mm}(\ell^{(a,d)}_t)^{(a,d)\in\mathcal{R}}_{a\leq t\leq d}} \hspace{2mm}\sum_{t=1}^T f\left(\ell^0_{t}+\sum_{(a,d)\in\mathcal{R}} \ell_{t}^{(a,d)}\right)\,,\\
&\text{s.t.}~\forall (a,d)\in\mathcal{R}\,,
\begin{cases}
\sum_{t=a}^d \ell_{t}^{(a,d)} = L^{(a,d)}/\delta\,,\\
\ell_{t}^{(a,d)}\geq 0\,,~\forall t\in\{a,\dots,d\}\,.\hspace{-5mm}
\end{cases}
\end{aligned}
\label{eq:offline}
\tag{$\mathcal{P}$}
\end{equation}

It is difficult to find an explicit charging scheduling algorithm solution of~\ref{eq:offline}.
However,~\ref{eq:offline} is a quadratic optimization problem (QP) and is easily solved by built-in Python function \textit{minimize} (in SciPy package), relying on a sequential least squares programming method.
As the objective function to minimize in~\ref{eq:offline} is strictly convex (because $f$ is), there is a unique minimal value $\tilde{G}$, but several possible optimal charging profiles $\tilde{\bm{\ell}}=(\ell^{(a,d)}_t)^{(a,d)\in\mathcal{R}}_{a\leq t\leq d}$ may exist.

In practice, such a scheduling may suffer from forecast errors made on arrival and departure time slots ``seen from" time slot 0 (i.e. the time when the problem has to be solved).
However, this (unrealistic) offline problem where all EV classes' demands are supposed to be known in advance can provide an upper bound for the performance of a more realistic method presented below, in order to measure its efficiency.

\section{Online two-step procedure}
\label{sec:online}


In this section, a more realistic assumption on the operators' access to information is studied.
Here, the operator does not know all the arrival and departure times in advance: the operator knows the arrival/departure time slots of an EV and its charging need only when the EV arrives at the EVCS (and communicates this information to the operator).
Therefore, for the whole time period $\mathcal{T}$, the operator waits for the next EV arrival to update charging scheduling decisions.
\subsection{Description}
At each EV arrival time slot $a\in\{1,\dots,T\}$ at the EVCS, the operator does the following procedure:
\begin{enumerate}
    \item Update the quantities $L^d_a$ left to charge from this arrival time slot $a$ to each possible departure time $d\in\{a,\dots,T\}$.
    For each $d\geq a$, $L^d_a$ is made of the charging need $L^{(a,d)}$ aggregated over EVs which arrived at $a$ and leave at $d$, plus the charging need left to charge of EVs which arrived earlier (and leave also at $d$).
    This charging need corresponds to the quantity $L^d_{a^-}$ which was left to charge from the previous EV arrival time slot $a^-$ (up to departure time $d$), minus the amount $\delta\times\sum_{t = a^-}^{a-1}\tilde{\ell}^d_{a^-,t}$ (defined in step~\ref{step:2}) that has already been charged since $a^-$:
    \begin{equation}
        \forall d \in \{a,\ldots,T\}, \quad
L^d_a= L^d_{a^-} - \hspace{1mm}\delta\hspace{-1mm}\sum_{t = a^-}^{a-1}\hspace{-1mm}\tilde{\ell}^d_{a^-,t} + L^{(a,d)}\,.
    \end{equation}
    Note that if no EV arrived before time slot $a$, the quantity $L^d_a$ is simply equal to $L^{(a,d)}$.
    The (increasingly) ordered set of departure times $d\in\{a,\dots,T\}$ where $L^d_a>0$ is denoted $\mathcal{D}_a$.
    Seen from instant a, it corresponds to all the (future) departure times for which a nonzero charging need has to be satisfied.

    \item \label{step:2}
     Use Algo.~\ref{algo:arrival} to compute the optimal value $\tilde{G}^a$ and per-class aggregated charging profile $\left(\tilde{\ell}^d_{a,t}\right)$ ($d\in\mathcal{D}_a,~a\leq t\leq d$),
    solutions of the following problem~\ref{eq:online}.
    This problem corresponds to the online charging scheduling of the per-class remaining energy needs $\bm{L}_a=\{L^d_a,~\forall d\in\mathcal{D}_a\}$ left to charge at arrival time $a$:
    \begin{equation}
    \begin{aligned}
            \min_{\hspace{7mm}(\ell^d_{a,t})^{d\in\mathcal{D}_a}_{a\leq t\leq d}} \hspace{2mm} \sum_{t=a}^T f\bigg( \ell^0_{t} + \sum_{d\in\mathcal{D}_a} \ell^d_{a,t} \bigg)\,,\\~\\~
                \end{aligned}
        \label{eq:online}
        \tag{$\mathcal{P}\left(\bm{L}_a\right)$}
    \end{equation}
    \vspace{-12mm}
            \begin{equation*}
            \hspace{3mm}
        \text{s.t. }\forall d\in\mathcal{D}_a,
        \begin{cases}
                \sum_{t=a}^d \ell^d_{a,t} = L^d_a/\delta\,,\\
     \ell^d_{a,t} \geq 0,~\forall t\in\{a,\dots,d\}\,.
             \end{cases}
    \end{equation*}
    Note that $\ell^d_{a,t}$ is the charging power \textit{programmed} for time slot $t$ and aggregated over all EVs which arrived at the EVCS at $a$ or before and leave at $d$: $\ell^d_{a,t}=\sum_{b = 1}^a \ell_t^{(b,d)}$.
    The charging power of these EVs at time $t$ may be updated later.
    \end{enumerate}

\begin{algorithm}[t]
\caption{(run at time $a$) Solution of~\ref{eq:online}}
\label{algo:arrival}
\nl\textbf{Available information:} departure times $\mathcal{D}_a$ and charging needs $L^d_a,~\forall d\in\mathcal{D}_a$\\
    \nl\For{each departure time $d\in\mathcal{D}_a$}
    {
   \nl Optimal charging profile of EVs leaving at $d$ using Water Filling solution $\bm{\ell}^{\text{WF}}$ (see Def.~\ref{def:standard}):
\begin{equation*}
(\tilde{\ell}^d_{a,t})_{a\leq t\leq d}
=\bm{\ell}^{\text{WF}}\left(L^{d}_a\,,\left(\ell_{t}^0+\sum_{u\in\mathcal{D}_a}^{u < d}\tilde{\ell}_{a,t}^u\right)_{a\leq t \leq d}\right)
\end{equation*}
    }
    \nl Minimal operator cost $\tilde{G}^a = \sum_{t=1}^{T} f\left(\ell^0_{t}+\sum_{d\in\mathcal{D}_a}\tilde{\ell}_{a,t}^d\right)$\\
\KwOut{Charging profiles $(\tilde{\ell}^d_{a,t})^{d\in\mathcal{D}_a}_{a\leq t\leq d}$ and cost $\tilde{G}^a$}
\end{algorithm}

\subsection{Analysis}
By following this two-step procedure for the whole time period $\mathcal{T}$, the operator minimizes at each arrival time slot $a$ the corresponding charging costs (objective function of~\ref{eq:online}). 
Note that the procedure only gives, at each time slot $t$ and for each departure time slots $d\geq t$, the optimal charging power aggregated over \textit{all} EVs leaving at $d$ and which arrived at $t$ or before.
This power is given by the last update $\tilde{\ell}^d_{\underline{t}, t}$ done at $\underline{t}$, the last EV arrival time slot before (or at) $t$.
Then, there are infinite ways to dispatch this power among the different EV classes $(a,d)$ with $a\leq \underline{t}$, so that $\sum_{a = 1}^{\underline{t}} \tilde{\ell}_t^{(a,d)}=\tilde{\ell}^d_{\underline{t},t}$.
The fact that at each arrival time slot $a$, Algo.~\ref{algo:arrival} provides an optimal solution to problem~\ref{eq:online} relies on the following definition.
\begin{defi}
If $f$ is increasing and convex and $\bm{\ell}^0$ is increasingly sorted, the vector solution of the following charging scheduling problem during time period $\mathcal{U}=\{t_i,\dots,t_f\}$:
\begin{equation}
\min_{\left(\ell_t\right)_{t\in\mathcal{U}}} \sum_{t \in\mathcal{U}} f\left(\ell^0_t+\ell_t\right)\,,\quad \text{s.t.}
\begin{cases}
\sum_{t\in\mathcal{U}} \ell_t = L/\delta\,,\\
\ell_t\geq 0\,,~\forall t\in\mathcal{U}\,,
\end{cases}
\tag{$\mathcal{S}$}
    \label{eq:standard}
\end{equation}
is:
\begin{equation}
\bm{\ell}^{\text{WF}}\left(L\,,\left(\ell_t^0\right)_{t\in\mathcal{U}}\right) = 
    \left\{\frac{L+L^0_{t_0}}{t_0(L)} - \ell^0_t\,, \quad \forall t\in\mathcal{U}\right\}\,,
    \label{eq:WF}
\end{equation}
where $L^0_t = \sum_{s\leq t}\ell^0_{t_i+s}$ and $t_0(L)\geq 1$ is such that $L\in ]\Delta_{t_0};\Delta_{t_0+1}]$, with $\Delta_t = t\times\ell^0_{t_i+t}-L^0_t$ for $t\leq t_f-t_i$ and $\Delta_{t_f-t_i} = +\infty$.
\label{def:standard}
\end{defi}
\noindent
This is a standard solution called Water-Filling (WF)~\cite{shinwari2012water, sohet20}.
In this solution, the total load (charging plus nonflexible consumption) at each time slot utilized for the charging operation is the same (it has the same "water level"\footnote{As would the water do by filling time slots with less nonflexible consumption.}), while the total load is higher in non-used time slots.
Equation~\eqref{eq:WF} shows that any increasing and convex function $f$ leads to the same optimal charging profile $\bm{\ell}^{\text{WF}}$, which smoothes as much as possible the total power load (made of the charging and nonflexible terms).
Also note that unlike the offline problem for which an optimization solver is needed, the solution has here an explicit form; it is obtained ``immediately".

The core idea of Algo.~\ref{algo:arrival} is to first solve (solution written $\tilde{\bm{\ell}}^{d_1}_a = (\tilde{\ell}^{d_1}_{a,t})_{a\leq t\leq d_1}$) the standard charging scheduling problem(~\ref{eq:standard}) introduced in Def.~\ref{def:standard} for EVs leaving the EVCS at the first departure time slot $d_1\in\mathcal{D}_a$, in function of the per-class aggregated charging need $L^{d_1}_a$ and the nonflexible vector $\bm{\ell}^0$.
Then, to solve this standard optimization problem(~\ref{eq:standard}) for EVs leaving the EVCS at the second departure time slot $d_2$, in function of $L^{d_2}_a$ and a fictitious nonflexible vector $\bm{\ell}^0 + \tilde{\bm{\ell}}^{d_1}_a$ which includes the charging profile of EVs which will have left earlier, and so on\dots

The following Prop.~\ref{prop:optimal} proves that Algo.~\ref{algo:arrival} gives an optimal solution of~\ref{eq:online} at each EV arrival time slot $a$.
Algorithm~\ref{algo:arrival} can actually be extended to problems with EV classes with presence time slots at the EVCS embedded in one another.
Note that the solution suggested in Algo.~\ref{algo:arrival} is not the only optimal charging profile $\tilde{\bm{\ell}}$ to give the unique minimal value $\tilde{G}^a$ of the corresponding charging cost function (e.g., the algorithmic solution with some charging power ``exchanged" between two EV classes and two time slots).

\begin{prop}
The output $\tilde{\bm{\ell}}^d_a \left(\forall d\in\mathcal{D}_a\right)$ of Algo.~\ref{algo:arrival} is a solution of optimization problem~\ref{eq:online}.
\label{prop:optimal}
\end{prop}

\begin{proof}
This can be shown by recurrence according to the departure time slots, using Def.~\ref{def:standard} and the Karush-Kuhn-Tucker conditions given that $f$ is convex and differentiable.
The complete proof is
\if\num1
given in Appendix.
\else
available online\footnote{\url{...}}.
\fi
\end{proof}

\subsection{Example}
\begin{figure}
    \centering
    \includegraphics[width = 0.5\textwidth]{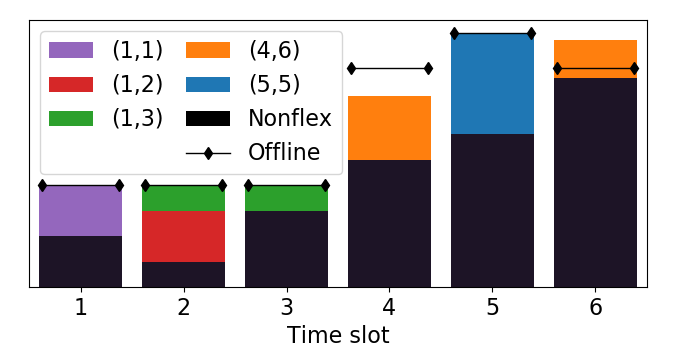}
    \caption{Example of the optimal charging profiles of five EV classes $(a,d)$ (colored bars), computed with the online scheduling charging problem of Sec.~\ref{sec:online} in function of a nonflexible consumption profile (black bars).
    \textit{The resulting total power load is less smooth than the one obtained by solving the offline scheduling problem~\ref{eq:offline} (diamond line), due to unexpected EV arrivals in the online procedure.}}
    \label{fig:example}
\end{figure}

The global online charging scheduling procedure is illustrated in Fig.~\ref{fig:example} with an example on a time period of $T=6$ time slots (e.g., the working hours from 8 am to 8 pm with $\delta = 2$ h) and five EV classes.
The operator starts by scheduling the charging profile of the EV classes which arrive for the first charging time slot: $(1,1)$, $(1,2)$, $(1,3)$.
Following Algo.~\ref{algo:arrival}, the operator starts with EV class $(1,1)$, which has no choice but to charge only during the first time slot.
Then, the operator charges EV class $(1,2)$ only during the second time slot because of the high total power load in the first slot, due to EV class $(1,1)$.
Finally, the charging need of EV class $(1,3)$ is adequately split between the first three time slots in order to smooth the total power load over the first three time slots. When other EVs arrive at the fourth time slot, the charging needs of EV classes which arrived before have already been fulfilled.
The operator plans to charge EV class $(4,6)$ during the fourth and fifth time slots.
Unfortunately, at time slot $t=5$, the operator must charge EV class $(5,5)$ which just arrived and has to postpone the charge of EV class $(4,6)$ to the sixth and last time slot.
Note that if the operator knew in advance that EVs would arrive at the fifth time slot, it could have charged more charging need of EV class $(4,6)$ during the fourth time slot, as in the offline charging problem (diamond line).

\section{Charging unit price (CUP)}
In addition to optimizing online the EV charging profiles, the operator of the EVCS sets a smart CUP (in price unit per energy unit) in order to indicate to EVs the actual (operator's) cost of their charging operation.
In this work, we suggest a smart CUP $\lambda^{(a,d)}$ for each EV class $(a,d)$ based on the charging costs the operator minimized by adequately choosing the per-class aggregated charging profiles of all EV classes. 
More precisely, we define $\lambda^{(a,d)}$ as the marginal operator costs corresponding to the charging need $L^{(a,d)}$ of EV class $(a,d)$, based on the Locational Marginal Pricing scheme~\cite{li2013distribution}.
Note that by definition, different EV classes may have different CUPs.
An interesting property of such a pricing scheme is that the EV class staying at the EVCS the whole time period $\mathcal{T}$ leads to a smaller marginal charging cost than an EV class staying only one time slot, because the former charging profile is more flexible than the latter (i.e. can be scheduled on a larger temporal period). Therefore, such a pricing mechanism can be used as an incentive for EV users to become more flexible for their charging operations.

In the offline charging scheduling problem introduced in Sec.~\ref{sec:offline}, the operator's charging cost $\tilde{G}$ is the one obtained by solving problem~\ref{eq:offline}, which gives the following CUPs:
\begin{equation}
    \lambda^{(a,d)}  = \frac{\partial \tilde{G}}{\partial L^{(a,d)}}\left(\bm{\ell}^0\,, \bm{L} \right)\,.
\end{equation}
By definition of the offline charging scheduling problem, the operator knows in advance all arrival and departure time slots and the corresponding per-class aggregated charging needs. Therefore the operator can compute the minimal charging cost $\tilde{G}$ by solving problem~\ref{eq:offline} and directly transmits the CUPs to EVs offline, before the whole time period $\mathcal{T}$.

In the online charging scheduling problem introduced in Sec.~\ref{sec:online}, the operator's charging cost considered to establish the CUP of EV class $(a,d)$ is $\tilde{G}^a$, the one computed at the arrival time slot $a$ of these EVs:
\begin{equation}
\lambda^{(a,d)} = \frac{\partial \tilde{G}^a}{\partial L^{(a,d)}}\left(\bm{\ell}^0\,, \textstyle\left(L_{a^-}^d\right)^{d\in\mathcal{D}_{a^-}}_{a^-\leq a}
\right)\,.
\label{eq:cup_online}
\end{equation}
Note that this cost $\tilde{G}^a$ may be different from the one when these EVs leave the EVCS, or from the one at the end of the whole charging operation at $t=T$:
if additional EVs arrive between $a$ and $d$, the operator updates the charging profiles and its costs with the online Algo.~\ref{algo:arrival}.
As mentioned in the introduction, the chosen pricing mechanism defined in Eq.~\ref{eq:cup_online} has the advantage of providing a price to EVs at their arrival, thus answering one of the main current EV users' expectations.

\section{Numerical results}
\subsection{Commuting framework with real data}

A natural use case corresponding to the charging scheduling problems introduced in this work is commuting.
Workers leave their EVs plugged in at an EVCS during working hours.
The EVCS is assumed to own PhotoVoltaic (PV) solar panels and use its PV generation to charge EVs and re-inject the remainder into the grid.
This PV generation is the only nonflexible term at the EVCS, and thus the vector $\bm{\ell}^0$ is nonpositive.
The data\footnote{Available at \url{https://www.renewables.ninja/}.} used for the PV generation comes from~\cite{PFENNINGER16} and represents the hourly generation of a 560 kilowatt peak during a random\footnote{The choice of the day does not affect the nature of the numerical results.} day (January 15, 2014) in Paris (see green curve in Fig.~\ref{fig:PV}).
The operator wants to minimize its charging costs by scheduling the EV charging during this day.

If at a time slot, there is more PV generation than EV charging, the operator is remunerated when re-injecting what is left of the PV generation into the grid.
However, as too much electricity re-injected may be potentially harmful
for the local distribution grid, this remuneration decreases with the quantity re-injected\footnote{See taxes on network companies: \url{https://bofip.impots.gouv.fr/bofip/797-PGP.html/identifiant=BOI-TFP-IFER-30-20210210} (in French).}.
If electricity from the grid is needed when EV total charging load exceeds PV generation, the operator's charging costs are often modeled in the literature by a quadratic proxy~\cite{mohsenian10}.
All this justifies the use of an increasing and quadratic function $f$ of the total load (either negative or positive) at each time slot in order to guide the operator's charging schedule.

\begin{figure}
    \centering
    \includegraphics[width = 0.5\textwidth]{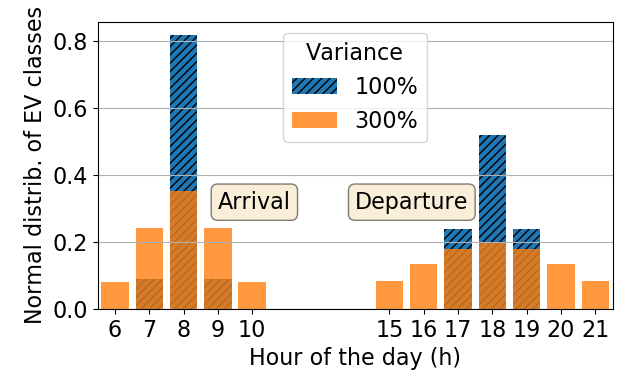}
    \caption{Arrival and departure discretized distributions, both with ENTD data and in the case where both variances were multiplied by three.
    \textit{The arrival distribution is more peaky than the departure one.}}
    \label{fig:distrib}
\end{figure}
The distribution of EVs in the different $(a,d)$ classes is given by the data from the French mobility survey ENTD\footnote{Enquête Nationale Transports et Déplacements: \url{https://www.statistiques.developpement-durable.gouv.fr/sites/default/files/2018-11/La\_mobilite\_des\_Francais\_ENTD\_2008\_revue\_cle7b7471.pdf} (in Fr.).} 2008.
The arrival and departure time slots are both modeled by independent normal distributions, respectively with means 8 am and 6 pm and variances 22 and 45 minutes (the arrival distribution is more peaky).
These distributions are discretized into time slots of one hour (following the PV generation data discretization) and shown in Fig.~\ref{fig:distrib}, with ENTD data and in the case where both variances were multiplied by three.
The latter scenario with higher variance could be realized with the remote working of nowadays.
We consider $N=100$ EVs, and the number $N^{(a,d)}$ of EVs in class $(a,d)$ is the product of $N$ with the distribution values of $a$ and $d+1$ (according to the convention that EV class $(a,d)$ can charge between the $a$-th and $d$-th time slots included, and leave at the beginning of time slot $d+1$).
EVs are assumed to have the same charging need, equivalent to their daily driving consumption: 6 kWh, due to the 30 km daily driving distance according to ENTD survey, at a 0.2 kWh/km average consumption per distance unit.
The charging need aggregated over class $(a,d)$ is therefore $L^{(a,d)}=6\times N^{(a,d)}$ kWh.

\begin{figure}
    \centering
    \includegraphics[width = 0.5\textwidth]{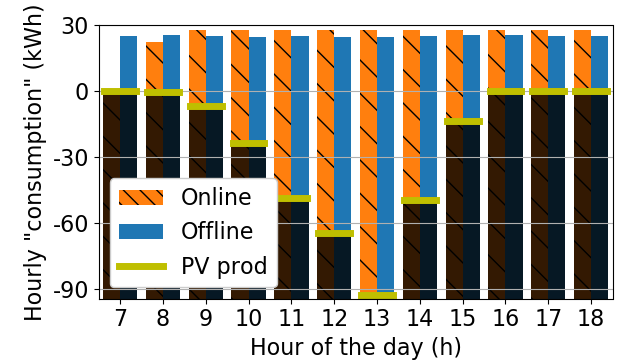}
    \caption{Comparison of optimal per-class aggregated charging profiles obtained with online and offline charging problems.
    \textit{In the online charging problem, the operator waits for a higher PV generation before charging the EVs which arrived at 7 am, while in the offline charging problem, the operator starts charging them right away because it knows at lot of EVs arrive at 8 am.}}
    \label{fig:PV}
\end{figure}
\begin{figure}
    \centering
    \includegraphics[width = 0.5\textwidth]{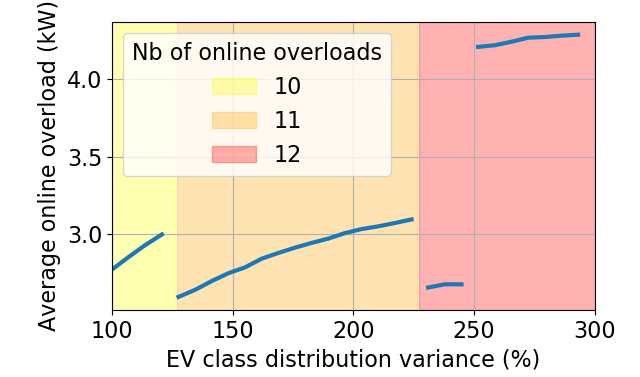}
    \caption{Average power overload of the online charging profile over the offline one and number of overload time slots in a day, in function of the EVs distribution variance.
    \textit{As the number of EVs arriving early increases with the variance, the overload increases too (see Fig.~\ref{fig:PV}), but may be divided into more time slots (when a new departure time is considered).}}
    \label{fig:depassement}
\end{figure}

\subsection{Comparison of online and offline charging profiles}
Figure~\ref{fig:PV} shows the optimal per-class aggregated charging profiles obtained with the online and offline charging problems and corresponding to the charging needs associated to ENTD data.
We can see that considering the online charging problem, the operator does not charge the few EVs which arrived at the EVCS at 7 am right away, but wait for time slots with higher PV generation.
In the offline charging problem, the operator knows that a lot of the PV generation will be used to charge the large number of EVs arriving at 8 am and therefore starts to charge the EVs arriving at 7 am as soon as possible.

Figure~\ref{fig:depassement} studies the power overload of the online charging profile with respect to the offline one.
More precisely, Fig.~\ref{fig:depassement} shows the number of time slots when the online charging power is greater than the offline one, and the average overload value during these time slots (blue line).
From 125~\%~of EVs distribution variance, some EVs start to leave at 8 pm from the EVCS (see Fig.~\ref{fig:distrib}) and thus the online overload (see Fig.~\ref{fig:PV}) can be divided into 11 time slots instead of 10, which mechanically reduces the average overload.
The same goes from 225~\%~of the variance, where some EVs start to leave at 9 pm.
However from 250~\%, some EVs start to arrive at 6 am at the EVCS which allows the offline charging scheduling to start one hour earlier while the online one still waits for the PV generation peak (see Fig~\ref{fig:PV}), hence the average overload increase.
Except from these discontinuities, the average overload increase with the variance for the same reasons: the higher the variance, the higher the nuber of EVs arriving early.

Finally, note that the explicit computations of the online charging profile are approximately a thousand times faster than a QP solver used for example for the offline optimization problem.


\subsection{Comparison of online and offline CUPs}
The last Fig.~\ref{fig:cups} compares the CUPs obtained with the online and offline methods.
To better illustrate the differences, the prices are plotted for different variances of EV arrival and departure distributions.
More precisely, we suppose that the variances of both the arrival and departure distributions can go up to 300~\%~of the ENTD values (see Fig.~\ref{fig:distrib}).
The vertical axis of Fig.~\ref{fig:cups} is normalized so that the highest point is equal to one.

First, Fig.~\ref{fig:cups} shows that from 250~\%~of variance values, the discretized EV classes distribution starts to consider EVs arriving at the EVCS at 6 am or 10 am (see corresponding online CUPs), which also explains the discontinuities in the CUPs.
Note that by definition, the online CUP $\lambda^{(a,d)}$ reflects the marginal cost of the operator computed at the arrival $a$ of the EV class $(a,d)$, and not the effective marginal cost (calculable at the departure $d$ of the EV class).
Therefore, in the online charging problem, EVs arriving at some time $a_1$ are likely to pay a CUP cheaper than EVs arriving at $a_2>a_1$, because the prices of the former only take into account the charging of EVs arriving at $a_1$, while the latter take both EV classes into account (see Fig.~\ref{fig:cups}).

This fairness aspect seems not addressed in the literature and could constitute a future work.
Similarly, most part of the increasing and decreasing features of the CUPs are also caused by this issue: for example, for a higher variance of arrival distribution, there is a lower proportion of EVs arriving before or at 8 am, which explains why the CUP associated to $a=8$ decreases.
Aside from that, Fig.~\ref{fig:cups} shows that the online CUPs do not depend on the departure time and the offline CUP is the same for all EV classes.
The reason is that, in this use case of commuting and PV generation, any small change in the charging need of an EV class can be compensated by the charging profiles of the other EV classes in order to keep a smooth total load (this is not true in the example given in Fig.~\ref{fig:example}).

\begin{figure}
    \centering
    \includegraphics[width = 0.5\textwidth]{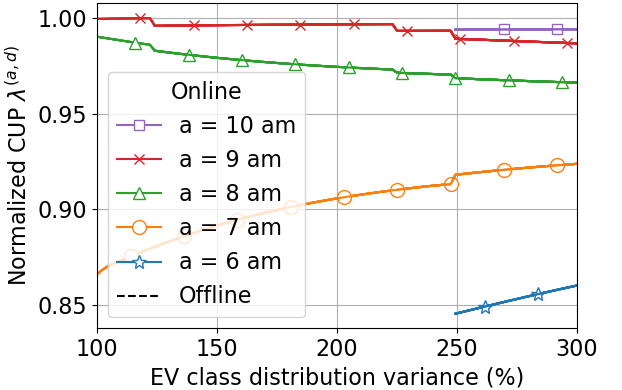}
    \caption{CUPs (vertical axis normalized) in function of the variance of EV classes' distribution, for each EV class $(a,d)$ and both online and offline scheduling problems.
    \textit{The offline CUP does not depend on the EV class due to the smoothed total load (see Fig.~\ref{fig:PV}).
    For the same reason, the online CUP does not depend on the EV departure time.
    However, because this CUP is given at the EV arrival, it is cheaper for EVs arriving earlier because the charging load of future EVs is not taken into account.}
    \if\num0
    \vspace{-4mm}
    \fi
}
    \label{fig:cups}
\end{figure}

\section{Conclusions and Perspectives}
This paper introduces an online charging scheduling algorithm adapted for asynchronous EVs arrival and departure.
At each EV arrival time slot, the operator of the EVCS updates the remaining quantities to charge until each potential departure time slot. Algorithm.\ref{algo:arrival} is used to find the corresponding optimal per-class aggregated charging profiles.
This online procedure requires minimal information and yields operator's charging costs only 1~\%~higher than the optimal value (obtained by an omniscient operator performing an offline scheduling optimization).
The CUPs defined as the optimized marginal operator's costs are guaranteed at EVs arrivals, but still suffer from fairness issues (first EVs to arrive likely to pay less).

The design of a CUP more correlated with the total time spent by an EV at the EVCS is currently under investigation.
The following working topic will be to integrate this CUP into the complete system taking into account the driving and charging decisions of EVs as well as the interactions between different system operators, and to show that this CUP can be an optimal incentive mechanism.

\if\num1
\appendices
\setcounter{secnumdepth}{0}
\section*{Appendix: Proof of Prop~\ref{prop:optimal}: Algo~\ref{algo:arrival} optimal}
\begin{proof}
Let $(\tilde{\ell}^d_{a,t})^{d\in\mathcal{D}_a}_{a\leq t\leq d}$ be the output of Algo~\ref{algo:arrival}.
The sorted departure times set can be written $\mathcal{D}_a = \{d_1,\dots,d_H\}$ with $H$ the set's cardinal.
Let $\mathcal{D}_a^n = \{d_1,\dots,d_n\}$ and $\bm{L}_{a,n}=\{L^d_a, ~\forall d\in\mathcal{D}_a^n\}$.
We are going to show that $P(n) = ``(\tilde{\ell}^d_{a,t})^{d\in\mathcal{D}_a^n}_{a\leq t\leq d}$ is solution of~$\mathcal{P}\left(\bm{L}_{a,n}\right)$" for all $n\in\{1,\dots,H\}$ by recurrence, which will prove Prop.~\ref{prop:optimal} because problems~\ref{eq:online} and $\mathcal{P}\left(\bm{L}_{a,H}\right)$ are equivalent.

\textit{Initialization:}
By definition, problems $\mathcal{P}\left(\bm{L}_{a,1}\right)$ and~\ref{eq:standard} are equivalent, therefore $(\tilde{\ell}^{d_1}_{a,t})_{a\leq t\leq d}=\bm{\ell}^{\text{WF}}\left(L_a^{d_1}\,,\left(\ell_{t}^0\right)_{a\leq t \leq d_1}\right)$ is solution of $\mathcal{P}\left(\bm{L}_{a,H}\right)$.

\textit{Recurrence:}
For any $n\in\{1,\dots,H-1\}$ we show $P(n+1)$, assuming $P(n)$.
Problem~$\mathcal{P}\left(\bm{L}_{a,n+1}\right)$ is convex and differentiable because function $f$ is, so that it is equivalent to its Karush-Kuhn-Tucker (KKT) conditions:
\begin{equation}
\forall d\hspace{-1mm}\in\hspace{-1mm}\mathcal{D}_a^{n+1},\forall t,~
\ell^d_{a,t} \times\bigg(\overbrace{f'\Big(\ell^0_{t} + \ell^{d_{n+1}}_{a,t}+ \hspace{-2mm}\sum_{d\in\mathcal{D}_a^{n}} \hspace{-2mm}\ell^d_{a,t}\Big) - \mu^d}^{\geq 0}\bigg)
=0,
\label{eq:KKT}
\end{equation}
with $\mu^d$ the (charging need) equality constraints Lagrange multipliers.
We show that $(\tilde{\ell}^d_{a,t})^{d\in\mathcal{D}_a^{n+1}}_{a\leq t\leq d}$ is solution of~\eqref{eq:KKT}.

By definition of $(\tilde{\ell}^{d_{n+1}}_{a,t})_{a\leq t\leq d}$ and the KKT conditions of~\ref{eq:standard}, \eqref{eq:KKT} is verified for $d=d_{n+1}$ and all $t\in\{a,\dots,d\}$, with $\mu^{d_{n+1}}=\mu_{n+1}$ the Lagrange multiplier of~\ref{eq:standard}.

Let $\mu_1,\dots,\mu_n$ be the Lagrange multipliers of~$\mathcal{P}\left(\bm{L}_{a,n}\right)$.
For $d_k\in\mathcal{D}_a^n$, there are two cases.
For $t$ such that $\ell_{a,t}^{d_{n+1}}=0$, \eqref{eq:KKT} is verified with $\mu^{d_k} = \mu_k$.
Otherwise, \eqref{eq:KKT} is verified with $\mu^{d_k} = \mu_{n+1}$.
\end{proof}
\fi

\bibliographystyle{ieeetr} 
\bibliography{myrefs}

\end{document}